\documentclass{article}
\usepackage[utf8]{inputenc}
\usepackage[T1]{fontenc}
\usepackage[english]{babel}
\usepackage{tikz}
\usepackage{subcaption}
\usepackage[shortcuts]{extdash}
\usepackage{diagbox}
\usepackage[inline]{enumitem}
\usepackage{multirow}
\usepackage{float}
\usepackage[square,numbers]{natbib}

\usepackage{graphicx}%
\usepackage{multirow}%
\usepackage{amsmath,amssymb,amsfonts}%
\usepackage{amsthm}%
\usepackage{mathrsfs}%
\usepackage[title]{appendix}%
\usepackage{xcolor}%
\usepackage{textcomp}%
\usepackage{manyfoot}%
\usepackage{booktabs}%
\usepackage{algorithm}%
\usepackage{algorithmicx}%
\usepackage{algpseudocode}%
\usepackage{listings}%
\usepackage{authblk}

\usetikzlibrary{decorations.pathmorphing,calc,intersections,arrows.meta}

\theoremstyle{definition}
\newtheorem{definition}{Definition}
\theoremstyle{plain}
\newtheorem{lemma}{Lemma}

\newtheorem{theorem}{Theorem}
\theoremstyle{remark}
\newtheorem{remark}{Remark}
\newtheorem{example}{Example}
\newtheorem{problem}{Problem}
\newtheorem{corollary}{Corollary}

\renewcommand{\Re}{\mathbb{R}}
\newcommand{\Sph}{\mathbb{S}}
\newcommand{\grad}{\nabla}

\DeclareMathOperator{\bd}{bd}
\DeclareMathOperator{\org}{org}
\DeclareMathOperator{\dest}{dest}
\DeclareMathOperator{\inter}{int}
\DeclareMathOperator{\conv}{conv}
\DeclareMathOperator{\Eq}{Eq}

\begin{document}

\title{Morse\==Smale complexes on convex polyhedra}

\author[1,2]{Bal\'azs Ludm\'any}
\author[1,3]{Zsolt L\'angi}
\author[1,4]{G\'abor Domokos}

\affil[1]{MTA-BME Morphodynamics Research Group, Budapest University of Technology and Economics, M\H uegyetem rakpart 1--3, Budapest, 1111, Hungary}
\affil[2]{Department of Control Engineering and Information Technology, Budapest University of Technology and Economics, Magyar Tud\'osok k\"or\'utja 2, Budapest, 1117, Hungary}
\affil[3]{Department of Algebra and Geometry, Budapest University of Technology and Economics, Egry J\'{o}zsef utca 1, Budapest, 1111, Hungary}
\affil[4]{Department of Mechanics, Materials \& Structures, Budapest University of Technology and Economics, M\H uegyetem rakpart 1--3, Budapest, 1111, Hungary}

\maketitle

\begin{abstract}
Motivated by applications in geomorphology, the aim of this paper is to extend Morse\==Smale theory from smooth functions to the radial
distance function (measured from an internal point), defining a convex polyhedron in $3$-dimensional Euclidean space. The resulting polyhedral Morse\==Smale complex may be regarded, on one hand, as a generalization of the Morse\==Smale complex of the smooth radial distance function defining a smooth, convex body, on the other hand, it could be also regarded as a generalization of the Morse\==Smale complex of the piecewise linear parallel distance function (measured from a plane), defining a polyhedral surface.  Beyond similarities, our paper also  highlights the marked differences between these three problems and it also relates our theory to other methods. Our work includes the design, implementation and testing of an explicit algorithm computing the Morse\==Smale complex on a convex polyhedron. 
\paragraph{keywords} Morse\==Smale complexes, polyhedral surfaces, static equilibrium points, radial function
\paragraph{MSC Classification} 57Q70, 52B70, 52B10
\end{abstract}


\section{Introduction}

Morse theory, the investigation of the critical points of a smooth function defined on a smooth manifold, has been in the focus of research since the middle of the 20th century \cite{milnor}. This area has received especially much attention since the seminal paper of Witten \cite{witten} directly relating Morse theory to quantum field theory, and has gained applications in many disciplines both within and outside mathematics (cf. e.g. \cite{uhlenbeck}). 
Following the rapid development of computer graphics, the description of smooth functions on a smooth manifold was extended in several directions:
\begin{itemize}
    \item The \emph{discrete problem} \cite{forman, gyulassy}, which was used to analyze the shape of objects of many different sizes from molecules \cite{cazals2003molecular}.
    \item The \emph{piecewise linear problem} represented, for example, by polyhedral surfaces over the $[x,y]$ plane.
    In \cite{banchoff_polyhedral, edelsbrunner_morse_smale} the polyhedral surface has been equipped with the so-called  \emph{height function}, which, in many applications, e.g. to terrains \cite{edelsbrunner_morse_smale},  everyday objects \cite{feng2013feature, dong2006spectral} is the distance measured from a fixed Euclidean plane $\Re^2$. The restriction of the height function to a polyhedral surface is piecewise linear, so all critical points of the function are located at vertices and the algorithm presented in \cite{edelsbrunner_morse_smale} relies on this fact. We will briefly refer to this case as the \emph{piecewise linear} problem.
    \item The \emph{point set problem}, discussed in  \cite{Edelsbrunner2003, dey2003flow, giesen2003flow}, where  the space $\Re^3$ has been equipped with the function whose value is the distance to the nearest of a finite set of  points.  The main applications of this approach are surface reconstruction and segmentation.Here, the resulting function is neither smooth nor piecewise linear, but the authors present a way to define its gradient vector field and flow curves. The finite point set also defines a Delaunay diagram and its dual Voronoi diagram, critical points of the distance function are the intersection points of these two diagrams. The  flow complex encapsulates the subsets of $\Re^3$ covered by the flow lines ending at the same critical points. We are going to call this case the \emph{point set} problem. While the flow complex complex proved to be useful in several applications, its potential is restricted: as noted in \cite{Edelsbrunner2003}, the concept of the flow complex can not be used to analyze discrete systems (such as \cite{forman, gyulassy}).
\end{itemize}
 As we will point out (see below and also Remark~\ref{rem:flow_cmp}), neither of the above theories is applicable to the construction of Morse theory on a compact, polyhedral surface which we will call the \emph{polyhedral problem}. In this paper our goal is to analyze the latter, by equipping it with the so-called
\emph{radial distance function}, i.e. the distance measured from a fixed, internal reference point. The radial distance function is, like the function in the point set problem, neither smooth nor piecewise linear, leading to significant differences in comparison with both the smooth and piecewise linear case. In particular, critical points may appear not only at vertices (as in the piecewise linear case) or at smooth points (as in the smooth case), but both at vertices and in the interiors of edges and faces of the simplicial complex.

Despite these differences, the piecewise linear and the polyhedral problem are also closely related: in the polyhedral (radial) case, if we let the reference point approach infinity, the radial distance function will approach the parallel distance function and all critical points will approach vertices. (We remark that the transition between the polyhedral and piecewise linear case is far from trivial: if  the reference point is not in the interior but at finite distance from the polyhedron, then neither the polyhedral nor the piecewise linear theories will apply.)

While the piecewise linear and the polyhedral problems are mathematically closely related, the applications driving these two models are somewhat different. While parallel distance functions can be efficient
models in image processing and serve as good approximations in local cartography \cite{edelsbrunner_morse_smale}, radial distance functions are the only option in case of particle shape models in geomorphology and planetology \cite{asteroid, natural, pebble, ludmany}. Also, for global cartography,  spherical geometry has to be taken into account.
 
The relationship between  the polyhedral problem and the point set problem is nontrivial at first sight: one might be tempted to regard the former as a special case of the latter. Indeed, by  reflecting the reference point to the plane of every face, the resulting Voronoi diagram has the polyhedron as one of its cells.  Nevertheless, a significant difference between the two problems is that, unlike in the point set problem, in the polyhedral problem the curves of the flow are restricted to remain on  boundary of the polyhedron. In Remark~\ref{rem:flow_cmp} we present an example exploiting this difference and showing that the resulting flow complex might be different from the Morse\==Smale complex we are after.

Our current paper is primarily motivated by particle shape modelling in geomorphology. Traditional geological shape descriptors
(e.g. axis ratios, roundness) have been recently complemented \cite{natural, pebble} by a new class, called \emph{mechanical shape descriptors} which proved to be rather efficient not only in determining the provenance of sedimentary particles \cite{Szabo2018} but also served surprising well in the identification of the shapes of asteroids \cite{natural,Langi2019}. Mechanical descriptors rely on the radial distance function measured from the center of mass to which we briefly refer as the mechanical distance function. The number of stable, unstable and saddle-type critical points are called \emph{first order} mechanical descriptors \cite{pebble, varkonyi_gomboc} whereas Morse\==Smale complexes belong to \emph{second order} mechanical descriptors \cite{natural, domokos_morse_smale}. While the mathematical properties of Morse\==Smale complexes  associated with a smooth mechanical distance function have been investigated in detail \cite{holmes, domokos_morse_smale}, so far there has been little experimental data to test this theory. The main obstacle in obtaining field data is the identification of Morse\==Smale complexes
on scanned images of particles \cite{ludmany}. These scans can be regarded as compact, polyhedral surfaces which we associate with the radial (mechanical) distance function.

In this paper we present the mathematical background to describe these complexes leading, in a natural way, to an algorithm to perform the above task, i.e. to identify Morse\==Smale complexes on the surfaces of polyhedra equipped with the radial distance function. We remark that a Morse\==Smale complex on a surface is naturally associated to the topological graph whose vertices and edges are those of the complex. This graph, called Morse\==Smale graph, often appears in applications \cite{holmes}, and by the above correspondence, information on one provides equivalent information on the other one.

The structure of the paper is as follows. In Section~\ref{sec:smooth} we give a brief summary of the elementary properties of Morse\==Smale complexes defined by a smooth convex body in $\Re^3$. The next two sections deal with the main goal of the paper: by investigating the properties of the radial distance function defined on a polyhedral surface, we build up a theory leading to the definition of a Morse\==Smale complex and the related Morse\==Smale graph on the surface (cf. Definitions~\ref{defn:MorseSmale} and \ref{defn:MSgraph}).  These definitions may appear, at least at first glance, rather similar to their smooth counterparts (Definitions~\ref{defn:MS_smooth} and
\ref{defn:smooth_graph}). Also, as we will see, Morse\==Smale graphs on polyhedra and Morse\==Smale graphs on smooth
convex bodies share most, but not all of their combinatorial properties. Still,  there also exist fundamental and significant differences between the smooth, the piecewise linear and the polyhedral problems, so each of these cases require different tools leading to these definitions. We present our mathematical results in two sections: in particular, in Section~\ref{subsec:local}, we find an analogue of the gradient vector field on the polyhedron, and in Section~\ref{subsec:global} we define ascending integral curves generated by this vector field and explore their properties, yielding the notion of Morse\==Smale complex on a polyhedral surface. Then in Section~\ref{sec:algorithm} we present an algorithm that computes the Morse\==Smale complex generated by a generic convex polyhedron. Finally, in Section~\ref{sec:remarks} we collect our additional remarks and open questions.

In the paper we denote by $\Re^3$ the $3$-dimensional Euclidean space. We regard this space as a $3$-dimensional real vector space equipped with the usual Euclidean norm, denoted by $| \cdot |$  or $N(\cdot)$. We denote the origin of this space by $o$, and the unit sphere centered at $o$ by $\Sph^2 = \{ x \in \Re^3 : |x| = 1 \}$. We use the notation $\inter (\cdot)$ and $\partial$ to the interior and the boundary of a set, respectively, and for points $x,y \in \Re^3$, we denote the closed segment with endpoints $x,y$ by $[x,y]$.
Let $K$ be a convex body (i.e. a compact, convex set with nonempty interior) in $\Re^3$, containing $o$ in its interior. The \emph{radial function $\rho_K : \Sph^2 \to (0,\infty)$} of $K$ (cf. \cite{schneider}) is defined by
\[
\rho_K (x) = \max \{ \lambda : \lambda x \in K \}.
\]
This function, often appearing in the literature, is the reciprocal of the gauge function of $K$ \cite{schneider}, and it clearly determines the convex body $K$ \cite{varkonyi_gomboc}. 
We note that while our considerations are valid if $o$ is an arbitrary interior point of $K$ chosen as reference point, in applications in mechanics and geology, it usually coincides with the center of gravity of $K$.
                                                                                       
\section{Morse\==Smale complex generated by a smooth convex body}\label{sec:smooth}

In this section we consider only convex bodies with $C^{\infty}$-class boundary.
In this case the mapping $u \mapsto \rho_K(u)$ is a $C^{\infty}$-diffeomorphism between $\Sph^2$ and $\partial K$, and hence, we assume that $\rho_K$ is infinitely many times differentiable on $\Sph^2$. This diffeomorphism establishes a natural correspondence between topological properties of the Morse\==Smale complex on $\Sph^2$ defined by $\rho_K$ and the Morse\==Smale complex on $\partial K$ defined by Euclidean norm in the sense that central projection from $\Sph^2$ to $\partial K$ maps the first complex to the second one, and vice versa, making possible to introduce these complexes in two equivalent ways. In this section we follow the first approach as it is more convenient for those familiar with Morse theory. Nevertheless, in Sections~\ref{subsec:local} and \ref{subsec:global}, since our considerations are usually based on the geometric properties of the convex polyhedron under investigation, we use the second approach.

\begin{definition}\label{defn:critical}
A point $x \in \Sph^2$ is called a \emph{critical point} of $\rho_K$, if the derivative of $\rho_K$ at $x$ is zero, that is if $(\grad \rho_K)(x) = 0$.
A noncritical point of $\rho_K$ is called a \emph{regular point} of $\rho_K$. A value $\alpha \in \Re$ is called a \emph{critical value} of $\rho_K$ if there is a critical point of $\rho_K$ satisfying $\rho_K(x) = \alpha$. A critical point $x$ is \emph{nondegenerate}, if the Hessian of $\rho_K$ at $x$ is not singular.
\end{definition}

Recall the well-known fact that if the Hessian of $\rho_K$ at a nondegenerate critical point $x$ has $0$, $1$ or $2$ negative eigenvalue, then $x$ is a local maximum, a saddle point or a local minimum of $\rho_K$, respectively.

\begin{remark}\label{rem:equilibrium}
If $x$ is a critical point of $\rho_K$, then the point $\rho_K(x) x \in \partial K$ is usually called a (static) equilibrium point of $K$ (with respect to $o$). Furthermore, for any nondegenerate critical point $x$ of $\rho_K$, the point $\rho_K(x) x$ is called a \emph{stable}, or a \emph{saddle-type} or an \emph{unstable} equilibrium point of $K$ if $x$ is a local minimum, a saddle point or a local maximum of $\rho_K$, respectively \cite{domokos_morse_smale}.
\end{remark}

\begin{definition}\label{defn:Morse_function}
The function $\rho_K$ is called a \emph{Morse function} if it is $C^{\infty}$-class, and all its critical points are nondegenerate. 
\end{definition}

In this paper we assume that $\rho_K$ is a Morse function. This implies that any critical point of $\rho_K$ has a neighborhood in $\Sph^2$ that does not contain any other critical point, and thus, by the compactness of $\Sph^2$, $\rho_K$ has only finitely many critical points.

\begin{definition} \label{def:integral}
An \emph{integral curve} $c : \Re \to \Sph^2$ is a curve, maximal with respect to inclusion, whose derivative at every point coincides with the value of the gradient of $\rho_K$ at that point; that is, $c'(t) = (\grad \rho_K)(c(t))$ for all $t \in \Re$. We call $\org(c) = \lim_{t \to -\infty} c(t)$ the \emph{origin}, and $\dest(c)=\lim_{t \to \infty} c(t)$ the \emph{destination} of $c$ \cite{edelsbrunner_morse_smale}.
\end{definition}

We remark that our smoothness conditions and the Picard-Lindel\"of Theorem yield that any regular point $x$ belongs to exactly one integral curve, all integral curves consists of only regular points, and two integral curves are either disjoint or coincide. Furthermore, the origin and the destination of every integral curve are critical points of $\rho_K$ (cf. e.g. \cite{arnold_counting, edelsbrunner_morse_smale}).

\begin{definition}\label{defn:smoothmanifolds}
The \emph{descending (resp. ascending) manifold} of a critical point $x$, denoted by $D(x)$ (resp. $A(x)$) is the union of $x$ and all integral curves with $x$ as their destination (resp. origin).
\end{definition}

The function $\rho_K$ is called \emph{Morse\==Smale} if all ascending and descending manifolds of $\rho_K$ intersect only transversally; or equivalently, if any pair of intersecting ascending and descending $1$-manifolds cross.

\begin{definition}\label{defn:MS_smooth}
The cells of the \emph{Morse\==Smale complex generated by $\rho_K$} are the sets obtained by intersecting a descending and an ascending manifold. The $2$-, $1$-, and $0$-dimensional cells of this complex are called \emph{faces, edges}, and \emph{vertices}, respectively.
\end{definition}

Observe that if $x \neq y$ are critical points of $\rho_K$, then $A(x) \cap D(x) = \{ x \}$, and $A(x) \cap D(y)$ is the union of all integral curves with origin $x$ and destination $y$.
Clearly, the vertices of the Morse\==Smale complex of $\rho_K$ are the critical points of $\rho_K$, and if $x$ is a stable and $y$ is an unstable point, then $A(x) \cap D(y)$ is an open set in $\Sph^2$, implying that it is a face of the complex. On the other hand, the same does not hold if $x$ or $y$ is a saddle point since every saddle point is the origin, and also the destination, of exactly two integral curves, corresponding to edges of the complex. It is well known (cf. e.g. Lemma 1 in \cite{edelsbrunner_morse_smale}) that every face of the Morse\==Smale complex is bounded by four edges, and the four critical points in the boundary are a stable, a saddle, an unstable and a saddle point, in this cyclic order, where the two saddle points may coincide.

\begin{definition}\label{defn:smooth_graph}
The topological graph $G$ on $\Sph^2$, whose vertices are the critical points of $\rho_K$, and whose edges are the edges of the Morse\==Smale complex, is called the \emph{Morse\==Smale graph} generated by $\rho_K$ \cite{holmes}. This graph is usually regarded as a \emph{$3$-colored} quadrangulation of $\Sph^2$, where the `colors' of the vertices are the three types of a critical point.
\end{definition}

\section[Generating the gradient vector field on a polyhedral surface: local properties of the function~$N_{\partial P}$]{Generating the gradient vector field on a polyhedral surface: local properties of the function~$\boldsymbol{N_{\partial P}}$}\label{subsec:local}

In the remaining part of the paper we deal with a convex polyhedron $P \subset \Re^3$ containing $o$ in its interior.
We note that the concepts of equilibrium points and nondegeneracy of convex polyhedra are already established in the literature \cite{balancing, Conway}. The following definition is from \cite{balancing}.

\begin{definition}\label{equilibrium}
Let $P \subset \Re^3$ be a convex polyhedron and let $o \in \inter (P)$.  We say that $q \in \partial P$ is an \emph{equilibrium point} of $P$ (with respect to $o$) if the plane $H$ through $q$ and perpendicular to $[o,q]$ supports $P$ at $q$. In this case $q$ is \emph{nondegenerate}, if $H \cap P$ is the (unique) face, edge or vertex of $P$ that contains $q$ in its relative interior; here, by the relative interior of a vertex we  mean the vertex itself. A nondegenerate equilibrium point $q$ is called \emph{stable, saddle-type} or \emph{unstable}, if $\dim (H \cap P) = 2,1$ or $0$, respectively. We call $P$ \emph{nondegenerate} if all its equilibrium points are nondegenerate. The points of $\partial P$ which are not equilibrium points are called \emph{regular points} of $\partial P$.
\end{definition}

In our investigation, we deal only with nondegenerate convex polyhedra. For any set $S \subset \Re^3$, we denote the restriction of the Euclidean norm function onto $S$ by $N_S$, and denote the affine hull of $S$ by $<S>$. Note that since $\partial P$ is piecewise smooth, the gradient of the function $N_{\partial P}$ exists only at interior points of the faces of $P$. To extend this notion to nonsmooth points of $\partial P$, we intend to find the `direction and rate of steepest ascent' for all nonsmooth points $q$ among all unit tangent vectors at $q$. The aim of Section~\ref{subsec:local} is to generalize the notion of gradient vector field by following this approach, and examine the properties of our generalization. In Section~\ref{subsec:global} we generalize the notion of integral curves and Morse\==Smale complex for our setting.

\begin{definition}\label{gradient}
Let $q \in \partial P$. If, for some vector $v \neq o$, the half line $\{ q + t v: t \geq 0 \}$ intersects $\partial P$ in a nondegenerate segment starting at $q$, we say that $v$ is a \emph{tangent vector} of $P$ at $q$. If $X$ is a face or an edge of $P$ containing $q$, and the gradient $(\grad N_{<X>})(q)$ is a tangent vector of $P$ at $q$, we say that this vector is a \emph{candidate gradient} at $q$.
\end{definition}

Note that by definition, if $v = (\grad N_{<X>}) (q)$ is a candidate gradient at $q$ and $v_0 = \frac{v}{|v|}$, then the one-sided directional derivative $\lim_{t \to 0^+} \frac{| q + tv_0|}{t} = |v|$ is strictly positive. Furthermore, the maximum of the one-sided directional derivatives of $N_{\partial P}$ at $q$ is the maximum of the lengths of the candidate gradients at $q$.
The main result of this section is the following.

\begin{theorem}\label{gradient_existence}
Let $q \in \partial P$. If $q$ is an equilibrium point of $P$, then there is no candidate gradient at $q$. If $q$ is a regular point of $P$, then there is at least one candidate gradient at $q$, and there is a unique candidate gradient at $q$ with maximal length.
\end{theorem}

\begin{proof}
We distinguish three cases.

\textbf{Case 1}: $q$ is an interior point of a face $F$ of $P$. If $q$ is the orthogonal projection of $o$ onto the plane $<F>$, then $q$ is clearly a stable equilibrium point of $P$, and there is no candidate gradient at $q$. In the opposite case there is a unique candidate gradient at $q$, namely $(\grad N_{<F>}) (q)$, which clearly has maximal length among all candidate gradients.

\textbf{Case 2}: $q$ is a relative interior point of an edge $E$ of $P$. Let $F_1$ and $F_2$ denote the two faces of $P$ that contain $E$, and for $i=1,2$, let $H_i$ denote the open half plane in $<F_i>$ that is bounded by the line $<E>$, and whose closure contains $F_i$. Furthermore, let $o_i$ denote the orthogonal projection of $o$ onto the plane $<H_i>$, and let $o_E$ be the orthogonal projection of $o$ onto $<E>$.

Observe that $o_i \in E$ yields, by the convexity of $P$, that $o_i$ is a degenerate equilibrium point of $P$, which contradicts our assumptions. Thus, we have $o_1, o_2 \notin E$. Similarly, if $o_1 \notin H_1$ and $o_2 \notin H_2$, then the coplanarity of $o,o_1, o_2, o_E$ implies that the dihedral angle of $P$ at $E$ is $(\pi-(o_1o_Eo) + (\pi-o_2o_Eo)) > \pi$, which contradicts the convexity of $P$. Thus, we have $o_i \in H_i$ for at least one value of $i$. Since the gradient of $N_{<H_i>}$ at any $x \neq o_i$ points in the direction of $x-o_i$, this implies that at most one of $(\grad N_{<F_i>}) (q)$ is a candidate gradient.

First, assume that $o_1 \in H_1$ and $o_2 \in H_2$. Then $(\grad N_{<F_i>}) (q)$ is not a candidate gradient for $i=1,2$. On the other hand, $(\grad N_{<E>}) (q)$ is a tangent vector of $P$ at $q$ if and only if $q \neq o_E$. Thus, this vector is the unique candidate gradient for $q \neq o_E$, and there is no candidate gradient at $q$ if $q=o_E$.

Now we consider the case that exactly one of the $o_i$s lies in $H_i$; say, we have $o_1 \in H_1$ and $o_2 \notin H_2$. Then, since $o_2 \in <E>$ yields $o_2 \in <E> \setminus E$, $(\grad N_{<F_2>}) (q)$ is a tangent vector of $P$ at $q$. On the other hand, in this case $|(\grad N_{<E>}) (q)|$, if it exists, is a directional derivative of the function $N_{<F_2>}$ at $q$, implying that $|(\grad N_{<E>}) (q)| \leq |(\grad N_{<F_2>}) (q) |$, with equality if and only if $o_2 \in <E> \setminus E$ and $(\grad N_{<E>}) (q) = (\grad N_{<F_2>}) (q)$.

\textbf{Case 3}: $q$ is a vertex of $P$. Let $F_1, F_2, \ldots, F_k$ denote the faces of $P$ containing  $q$ such that $k \geq 3$, and for $i=1,2,\ldots,k$, the set $E_i= F_i \cap F_{i+1}$ is an edge of $P$. Then the edges of $P$ containing $q$ are $E_1, E_2, \ldots, E_k$. First, observe that if $q$ is an equilibrium point of $P$, then for every tangent vector $v$ of $P$ at $q$, the one-sided directional derivative of $N_{\partial P}$ in the direction of $v$ is negative, which yields that there is no candidate gradient at $q$. Thus, we may assume that $q$ is a regular point of $P$, which yields that $(\grad N_{<E_i>})(q)$ is a candidate gradient for at least one value of $i$.

We show that there is at most one value of $i$ such that $(\grad N_{<F_i>})(q)$ is a candidate gradient. Indeed, suppose for contradiction that $(\grad N_{<F_i>})(q)$ and $(\grad N_{<E_j>})(q)$ are candidate gradients at $q$, where $i \neq j$. Let $E = <F_i> \cap <F_j>$, and let $H_i$ (resp. $H_j$) denote the open half plane in $<F_i>$ (resp. $<F_j>$) that is bounded by the line $E$, and whose closure contains $F_i$ (resp. $F_j$). Similarly like in Case 2, we may observe that the fact that $(\grad N_{<F_i>})(q)$ and $(\grad N_{<F_j>})(q)$ are tangent vectors of $P$ at $q$ yields that the orthogonal projection of $o$ onto $<F_i>$ (resp. $<F_j>$) is not contained in $H_i$ (resp. $H_j$), and thus, the dihedral angle between $<\{o\} \cup E>$ and $H_i$ (resp. $H_j$) is at least $\frac{\pi}{2}$. But this contradicts the fact that the intersection of the two supporting half spaces of $P$, containing $F_i$ and $F_j$, respectively, is a convex polyhedral region.

Consider the case that for some value of $i$, $v=(\grad N_{<F_i>})(q)$ is a candidate gradient, i.e. it is a tangent vector of $P$ at $q$. Then, clearly, the convexity of $P$ implies that for any unit tangent vector $v' \neq \frac{v}{|v|}$ of $P$ at $q$, the one-sides directional derivative of $N_{\partial P}$ at $q$ in the direction of $v'$ is strictly less than $|v|$. Thus, in this case $v=(\grad N_{<F_i>})(q)$ has maximal length among all candidate gradients at $q$.

Finally, we deal with the case that $(\grad N_{<F_i>})(q)$ is not a tangent vector of $P$ at $q$ for any value of $i$, but there are more than one values of $i$ such that $(\grad N_{<E_i>})(q)$ is a candidate gradient of maximal length. Let $i_1 \neq i_2$ be such values. For $j=1,2$, consider the closed half line $R_{i_j} = \{ q+ t v_{i_j} : t \geq 0 \}$ where $v_{i_j}$ is the unit vector pointing from $q$ towards the other endpoint of $E_{i_j}$, and let $H= <E_{i_1} \cup E_{i_2}>$. Note that $|(\grad N_{<E_{i_j}>})(q)|$ is the directional derivative of $N_H$ in the direction of $v_{i_j}$, and hence, the equality $|(\grad N_{<E_{i_1}>})(q)| = |(\grad N_{<E_{i_2}>})(q)|$ implies that $(\grad N_H)(q)$ points in the direction of $v_{i_1}+v_{i_2}$. On the other hand, the half line $\{ q + t(v_{i_1}+v_{i_2}) : t \geq 0 \}$ intersects $P$ in a nondegenerate segment, and thus, there is a unit tangent vector $v$ of $P$ at $q$ such that the one-sided directional derivative of $N_{\partial P}(q)$ in the direction of $v$ is strictly greater than $|(\grad N_{<E_{i_j}>})(q)|$ ($j=1,2$), contradicting the assumption that the latter two vectors have maximal length.
\end{proof}

\begin{remark}\label{rem:connected_directions}
By the argument in the last paragraph, for any vertex $q$ of $P$, the set of tangent vectors $v$ of $P$ at $q$ with the property that in the direction of $v$ the one-sided derivative of the distance function increases is connected, implying the same statement for the set of tangent vectors in which the distance function decreases.
\end{remark}

\begin{definition}\label{defn:extended_gradient}
If $q$ is a regular point of $P$, then we call the unique candidate gradient at $q$ with maximal length the \emph{extended gradient} of $N_{\partial P}$ at $q$. If $q$ is an equilibrium point of $P$, we say that the extended gradient of $N_{\partial P}$ at $q$ is $0$. We denote the extended gradient of $N_{\partial P}$ at $q$ by $(\grad^\text{ext} N_{\partial P}) (q)$.
\end{definition}

Our analysis in Case 2 of the proof of Theorem~\ref{gradient_existence} leads to Definition~\ref{crossed}.

\begin{definition}\label{crossed}
Let $E$ be an edge of $P$, and let $F_1$ and $F_2$ be the two faces of $P$ containing $E$. For $i=1,2$, let $H_i$ denote the open half plane in $<F_i>$ that is bounded by the line $<E>$, and whose closure contains $F_i$, and let $o_i$ be the orthogonal projection of $o$ onto $<F_i>$. If $o_1 \in H_1$ and $o_2 \in H_2$, then for any regular point $q$ in the relative interior of $E$, the extended gradient at $q$ is parallel to $E$. In this case we say that $E$ is a \emph{followed edge} of $P$. If either $o_1 \notin H_1$ or $o_2 \notin H_2$, then we say that $E$ is a \emph{crossed} edge. In particular, if $E$ is a crossed edge with $o_1 \notin H_1$, then for every point $q$ in the relative interior of $E$, the extended gradient at $q$ points towards the interior of $F_1$. In this case we say that $E$ is \emph{crossed from $F_2$ to $F_1$}.
\end{definition}

\begin{figure}[h]
  \centering
  \begin{subfigure}{.45\textwidth}
    \centering
    \begin{tikzpicture}
      \draw (0,1) -- (-2,0) -- (0,-1) -- (2,0) -- cycle;
      \draw (-1.2,0.7) node {$F_2$};
      \draw (1.2,0.7) node {$F_1$};
      
      \draw[ultra thick] (0,1) -- (0,-1);
      \draw (0,1.2) node {$E$};
      
      \fill (-2,-1.5) circle (0.05);
      \begin{scope}
        \clip (0,1) -- (-2,0) -- (0,-1) -- cycle;
        \begin{scope}[xshift=-2cm, yshift=-1.5cm]
          \foreach \layer in {1,2,...,6} {
            \pgfmathsetmacro\first{40/\layer};
            \pgfmathsetmacro\last{360-\first};
            \foreach \angle in {0,\first,...,\last} {
              \draw[-latex] (\angle:0.5*\layer-0.4) -- ++(\angle:0.3);
            }
          }
        \end{scope}
      \end{scope}
      \draw (-1.8,-1.7) node {$o_2$};
      
      \fill (1,-1.3) circle (0.05);
      \begin{scope}
        \clip (0,-1) -- (2,0) -- (0,1) -- cycle;
        \begin{scope}[xshift=1cm, yshift=-1.3cm]
          \foreach \layer in {1,2,...,6} {
            \pgfmathsetmacro\first{40/\layer};
            \pgfmathsetmacro\last{360-\first};
            \foreach \angle in {0,\first,...,\last} {
              \draw[-latex] (\angle:0.5*\layer-0.4) -- ++(\angle:0.3);
            }
          }
        \end{scope}
      \end{scope}
      \draw (0.8,-1.5) node {$o_1$};
    \end{tikzpicture}
    \caption{Followed edge}
  \end{subfigure}
  \begin{subfigure}{.45\textwidth}
    \centering
    \begin{tikzpicture}
      \draw (0,1) -- (-2,0) -- (0,-1) -- (2,0) -- cycle;
      \draw (-1.2,0.7) node {$F_2$};
      \draw (1.2,0.7) node {$F_1$};
      
      \draw[ultra thick] (0,1) -- (0,-1);
      \draw (0,1.2) node {$E$};
      
      \fill (-2,-1.5) circle (0.05);
      \begin{scope}
        \clip (0,1) -- (-2,0) -- (0,-1) -- cycle;
        \begin{scope}[xshift=-2cm, yshift=-1.5cm]
          \foreach \layer in {1,2,...,6} {
            \pgfmathsetmacro\first{40/\layer};
            \pgfmathsetmacro\last{360-\first};
            \foreach \angle in {0,\first,...,\last} {
              \draw[-latex] (\angle:0.5*\layer-0.4) -- ++(\angle:0.3);
            }
          }
        \end{scope}
      \end{scope}
      \draw (-1.8,-1.7) node {$o_2$};
      
      \fill (-1,-1.3) circle (0.05);
      \begin{scope}
        \clip (0,-1) -- (2,0) -- (0,1) -- cycle;
        \begin{scope}[xshift=-1cm, yshift=-1.3cm]
          \foreach \layer in {1,2,...,6} {
            \pgfmathsetmacro\first{40/\layer};
            \pgfmathsetmacro\last{360-\first};
            \foreach \angle in {0,\first,...,\last} {
              \draw[-latex] (\angle:0.5*\layer-0.4) -- ++(\angle:0.3);
            }
          }
        \end{scope}
      \end{scope}
      \draw (-0.8,-1.5) node {$o_1$};
    \end{tikzpicture}
    \caption{Crossed edge}
  \end{subfigure}
  \caption{Possible types of edges. The extended gradient field on faces is also shown.}
\end{figure}
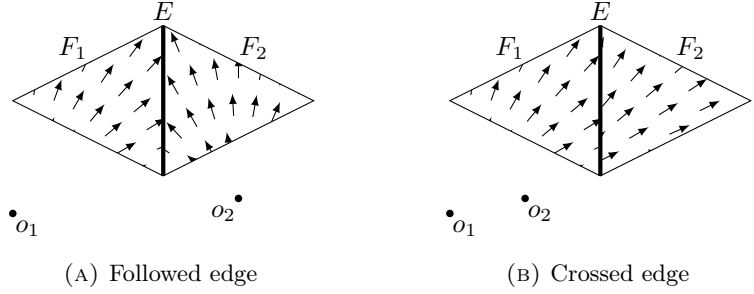

We need the following remark in Section~\ref{subsec:global}.

\begin{remark}\label{rem:cyclic_crossed}
Let $q$ be a vertex of  $P$ and let $F_1,F_2, \ldots, F_k$ be the faces of $P$ containing $q$ such that for $i=1,2,\ldots,k$, $E_i = F_{i-1} \cap F_i$ is an edge of $P$. Assume that all edges $E_i$ are crossed edges. Then there is some value of $i$ such that $E_i$ is crossed from $F_i$ to $F_{i+1}$, and a value of $i$ such that $E_i$ is crossed from $F_{i+1}$ to $F_i$. Indeed, if $d_i$ denotes the distance of $o$ and $<F_i>$, then the fact that $E_i$ is crossed from $F_i$ to $F_{i+1}$ implies that $d_i < d_{i+1}$. Thus, if $t$ is a value of $i$ such that $d_i$ is minimal, then $E_{t-1}$ is crossed from $F_t$ to $F_{t-1}$, and $E_t$ is crossed from $F_t$ to $F_{t+1}$.
\end{remark}

\begin{remark}\label{rem:extended_gradient}
As a summary, we collect the following rules to find the extended gradient at any given boundary point $q$ of $P$.
\begin{itemize}
\item If $q$ is an interior point of a face $F$, then $(\grad^\text{ext} N_{\partial (P)}) (q) = (\grad N_{<F>}) (q)$.
\item If $q$ is a relative interior point of an edge $E$ of $P$, where $F_1$ and $F_2$ denotes the two faces of $P$ containing $E$, then:
  \begin{itemize}
  \item if $E$ is a followed edge, then $(\grad^\text{ext} N_{\partial (P)}) (q) = (\grad N_{<E>}) (q)$;
  \item if $E$ is a crossed edge from $F_1$ to $F_2$, then $(\grad^\text{ext} N_{\partial (P)}) (q) = (\grad N_{<F_2>}) (q)$.
  \end{itemize}
\item If $q$ is a vertex of $P$, then:
  \begin{itemize}
	\item if $q$ is an unstable point of $P$, then $(\grad^\text{ext} N_{\partial (P)}) (q) = 0$;
  \item if there is a face $F$ containing $q$ such that $(\grad N_{<F>}) (q)$ is a tangent vector of $P$ at $q$, then $(\grad^\text{ext} N_{\partial (P)}) (q) = (\grad N_{<F>}) (q)$;
	\item in the remaining case there is a followed edge $E$ with endpoint $q$ such that $(\grad^\text{ext} N_{\partial (P)}) (q) = (\grad N_{<E>}) (q)$.
  \end{itemize}
\end{itemize}
It is also worth noting that for any face $F$, the gradient vector $(\grad N_{<F>})(q)$ is a positive scalar multiple of the vector pointing from the orthogonal projection of $o$ onto $<F>$ to $q$.
\end{remark}

The consequence of Remark~\ref{rem:extended_gradient} is that after the classification of edges as crossed or followed, we can always determine which candidate gradient is the extended gradient. There is no need to compute more than one candidate, which makes our algorithm quicker and also less prone to numerical errors.

\section[Morse\==Smale complex on a polyhedral surface: global properties of the function~$N_{\partial P}$]{Morse\==Smale complex on a polyhedral surface: global properties of the function~$\boldsymbol{N_{\partial P}}$}\label{subsec:global}

So far, we have defined a variant of the gradient vector field on the boundary of $P$. This vector field, as already $\partial P$, is not smooth, implying that the definition of integral curve cannot be applied to our case without any change. Our definition is as follows, where we denote the interval $(-\infty,0)$ by $\Re^-$.

\begin{definition}\label{defn:integral_curve}
An \emph{ascending curve} $c: \Re^-\to \partial P$ is a curve maximal with respect to inclusion, such that for any point $q=c(\tau)$ of $c$ the right-hand side derivative of $c$ is equal to the extended gradient at $q$, namely $c_+'(\tau)=(\grad^\text{ext} N_{\partial (P)}) (q)$. We call \(\lim_{\tau\to-\infty}c(\tau)\) the \emph{origin} and \(\lim_{\tau\to 0^-}c(\tau)\) the \emph{destination} of the ascending curve $c$.
\end{definition}

As we will see, the origin (resp. destination) of any ascending curve is a stable or a saddle point (resp. a saddle or an unstable point) of $P$. Thus, for any ascending curve $c$ with origin $s$, the first segment of $c$ is a segment on a line perpendicular to $[o,s]$, and, in particular, the derivative $c'(\tau)$ tends to zero as the point tends to $s$. From this, an elementary computation shows that the curve $c$ is parameterized on an interval unbounded from below. On the other hand, if $s$ is the destination of $c$, then the last segment in $c$ is one not perpendicular to $[o,s]$. Thus, $c$ is parameterized on an interval bounded from above. This shows that every ascending curve can be parameterized on the interval $\Re^-$ in a unique way.

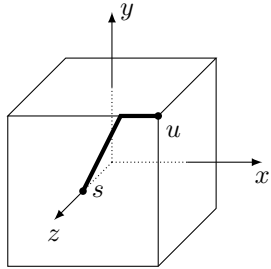
\begin{figure}[h]
  \centering
  \begin{tikzpicture}
    \draw[densely dotted] (0,0,0) -- (1,0,0);
    \draw[-latex] (1,0,0) -- (2,0,0);
    \draw[densely dotted] (0,0,0) -- (0,1,0);
    \draw[-latex] (0,1,0) -- (0,2,0);
    \draw[densely dotted] (0,0,0) -- (0,0,1);
    \draw[-latex] (0,0,1) -- (0,0,2);
    
    \draw (-1,1,-1) -- (-1,1,1) -- (1,1,1) -- (1,1,-1) -- cycle;
    \draw (-1,1,1) -- (-1,-1,1) -- (1,-1,1) -- (1,-1,-1) -- (1,1,-1);
    \draw (1,1,1) -- (1,-1,1);
    \draw[ultra thick] (0,0,1) -- (0.5,1,1) -- (1,1,1);
    \fill (0,0,1) circle (0.05);
    \fill (1,1,1) circle (0.05);
    \draw (0.2,0,1) node {$s$};
    \draw (1.2,0.8,1) node {$u$};
    \draw (2,-0.2,0) node {$x$};
    \draw (0.2,2,0) node {$y$};
    \draw (0,-0.2,2) node {$z$};
  \end{tikzpicture}
  \caption{One of the infinitely many ascending curves between the stable point $s$ and the unstable point $u$ on a cube}
\end{figure}

To study the properties of the ascending curves, we start with some preliminary observations.

\begin{remark}\label{rem:cover}
Clearly, any point of an ascending curve is regular. On the other hand,
to any regular point $q \in \partial P$, one can assign a `direction of greatest descent' as a direction in which the one-sided directional derivative $\lim_{t \to 0^+} \frac{| q + tv_0|}{t} = |v|$ is minimal among all unit tangent vectors $v_0$ of $P$ at $q$. Similarly like in case of extended gradient, it can be shown that this direction exists, with the corresponding derivative being strictly negative, and also that it is the opposite of a candidate gradient $(\grad N_{<X>}) (q)$ for some face or edge $X$ of $P$ containing $q$. This implies that every regular point of $P$ belongs to the relative interior of at least one ascending curve of $P$. Note that unlike in the smooth case, a regular point may belong to more than one ascending curve. This is true, for example, for any regular point in the relative interior of a followed edge, where, in general, three ascending curves meet.
\end{remark}

\begin{remark}\label{rem:merge}
By their definition, no two ascending curves can cross or split. On the other hand, as we have seen in Remark~\ref{rem:cover}, they can merge (cf. Figure~\ref{fig:merge}).
\end{remark}

\begin{figure}[h]
  \centering
  \begin{tikzpicture}
    \fill (0,0) circle (0.08);
    \fill[gray] (0,0.5) circle (0.08);
    \draw (0,1) -- (-2,0) -- (0,-1) -- (2,0) -- cycle;
    \draw (0,-1) -- (0,1);
    \begin{scope}
      \clip (0,1) -- (-2,0) -- (0,-1) -- (2,0);
      \draw[ultra thick, gray] (2,-1) -- (0,0.5) -- (0,1);
      \draw[very thick, densely dashed] (-2,-1) -- (0,0) -- (0,1);
    \end{scope}
    \node at (0.2,-0.6) {$E$};
    \node at (-0.5,0) {$c_1$};
    \node[gray] at (1,0) {$c_2$};
    \node at (0.3,0) {$p_1$};
    \node[gray] at (-0.3,0.5) {$p_2$};
  \end{tikzpicture}
  \caption{Two ascending curves $c_1$ and $c_2$ merging along a followed edge $E$ at $p_2$}\label{fig:merge}
\end{figure}
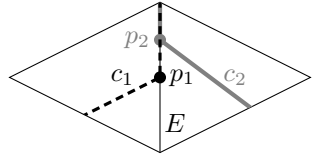

Note that every ascending curve is the union of some segments in the interiors of faces or on followed edges of $P$.
Our next theorem is an important tool to describe the geometric properties of such a curve.

\begin{theorem}\label{thm:notmultiple}
Let $S_1,S_2, \ldots,S_k$ be a sequence of consecutive segments in an ascending curve such that $S_i$ lies in the face or followed edge $X_i$ of $P$, where $X_i \neq X_{i+1}$. Let $d_i$ denote the distance of $<X_i>$ from $o$. Then $d_1,d_2,\ldots, d_k$ is a strictly increasing sequence.
\end{theorem}

\begin{proof}
For $i=1,2,\ldots,k-1$, let $H_i$ denote the plane spanned by $S_i \cup S_{i+1}$, where the existence of $H_i$ follows from the condition that $S_i$ and $S_{i+1}$ are not contained in the same face or edge. For $i=1,2\ldots,k$, let $o_i$ denote the orthogonal projection of $o$ onto $H_i$, and set $S_i = [q_{i-1},q_i]$. Let $C_i$ be the convex angular region bounded by the two half lines starting at $q_i$ and containing $S_i$ and $S_{i+1}$, respectively. By the definition of extended gradient, the plane through $<S_i>$ and perpendicular to the plane $<\{o \} \cup S_i>$ supports $P$. This implies that $o_i$ and $q_{i+1}$ lie in the same open half plane of $H_i$ bounded by $<S_i>$. We obtain similarly that $o_i$ and $q_{i-1}$ lie in the same open half plane of $H_i$ bounded by $<S_{i+1}>$. Thus, we have that $o_i$ is an interior point of $C_i$.

We show that for any value of $i$, the distance of $<S_i>$ from $o$ is equal to $d_i$. Indeed, this is trivial if $X_i$ is a followed edge. On the other hand, if $X_i$ is a face, then our observation follows from the fact that by the properties of extended gradient, the orthogonal projection $o_i$ of $o$ onto $<X_i>$ lies on $<S_i>$.

Now we prove the assertion. Let $1 \leq i \leq k-1$. Since $S_i$ and $S_{i+1}$ are consecutive segments on an ascending curve, moving a point $q$ at a constant velocity from $q_{i-1}$ to $q_{i+1}$ along $S_i \cup S_{i+1}$, the distance of $q$ from $o$ strictly increases. By the Pythagorean Theorem, the same property holds for the distance of $q$ and $o_i$. Thus, $q_i$ does not separate the orthogonal projection of $o_i$ onto $<S_i>$ from $q_{i-1}$, while it strictly separates the orthogonal projection of $o_i$ onto $<S_{i+1}>$  from $q_{i+1}$. Since $o_i \in C_i$, this yields that $o_i$ is strictly closer to $<S_i>$ than to $<S_{i+1}>$ (Figure~\ref{fig:segments} illustrates this arrangement). Thus, the inequality $d_i < d_{i+1}$ follows from the observation in the previous paragraph.

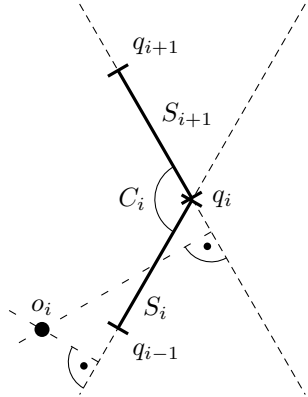
\begin{figure}[h]
  \centering
  \begin{tikzpicture}
    \draw[densely dashed] (120:3) -- (300:3);
    \draw[very thick, |-|] (0,0) -- (120:2);
    \draw (120:1)+(30:0.5) node {$S_{i+1}$};
    \draw (120:2)+(30:0.6) node {$q_{i+1}$};
    
    \draw[densely dashed] (60:3) -- (240:3);
    \draw[very thick, |-|] (0,0) -- (240:2);
    \draw (240:1.5)+(330:0.3) node {$S_i$};
    \draw (240:2)+(330:0.6) node {$q_{i-1}$};
    
    \draw (120:0.5) arc[start angle=120, end angle=240, radius=0.5];
    \draw (180:0.8) node {$C_i$};
    \draw (0:0.4) node {$q_i$};
    
    \draw[loosely dashed, name path=orthogonal 1] (300:0.5) -- ++(210:3);
    \draw (300:0.9) arc[start angle=300, end angle=210, radius=0.4];
    \fill (300:0.5)+(255:0.2) circle (0.05);
    
    \draw[loosely dashed, name path=orthogonal 2] (240:2.5) -- ++(150:1.5);
    \draw (240:2.9) arc[start angle=240, end angle=150, radius=0.4];
    \fill (240:2.5)+(195:0.2) circle (0.05);
    
    \fill[name intersections={of=orthogonal 1 and orthogonal 2}] (intersection-1) circle (0.1);
    \draw[name intersections={of=orthogonal 1 and orthogonal 2}] (intersection-1)+(90:0.3) node {$o_i$};
  \end{tikzpicture}
  \caption{Segments \(S_i\) and \(S_{i+1}\) of an ascending curve.}\label{fig:segments}
\end{figure}
\end{proof}

In the following we give an upper bound for the number of segments in any ascending curve.

\begin{theorem}\label{thm:polygonal}
Every ascending curve is a polygonal curve whose origin and destination are equilibrium points of $P$, and it intersects every face and edge of $P$ in at most one segment. 
\end{theorem}

\begin{proof}
First, we show that every ascending curve is rectifiable. To do it, it is sufficient to observe that the intersection of every face $F$ of $P$ with an ascending curve is a union of segments contained on lines passing through the orthogonal projection $o_F$ of $o$ onto $<F>$. Since the Euclidean norm $N$ strictly increases along an ascending curve, this implies that the total length of these segments is not more than the distance of the farthest point of $F$ from $o_F$.

Next, we show that every ascending curve is the union of finitely many segments. To do it, it is sufficient to show that it intersects the interior
of every face of $P$ in finitely many segments. Suppose for contradiction that there is an ascending curve $c$ and a face $F$ such that $\inter(F) \cap c(\Re^-)$ has infinitely many components. Since $c(\Re^-)$ is rectifiable, this implies that $F$ has a vertex $q$ such that every neighborhood of $q$ contains infinitely many components of $\inter(F) \cap c(\Re^-)$. Also by rectifiability, this implies that for some closed interval $I$, the arc $c(I)$ starts at an interior point of $F$, crosses each edge of $P$ starting at $q$, and ends up at an interior point of $F$. Note that this implies that each edge of $P$ starting at $q$ is a crossed edge, and the direction of the crossing corresponds to the cyclic order of the faces around $q$. But this contradicts Remark~\ref{rem:cyclic_crossed}.

Finally, the fact that an ascending curve intersects every face and edge in at most one segment follows from Theorem~\ref{thm:notmultiple} and the fact that it is a polygonal curve.
\end{proof}

Now we are ready to define the ascending and the descending manifolds.

\begin{definition}\label{defn:ascending_manifold}
  The \emph{ascending (resp. descending) polyhedral manifold} of an equilibrium point $x$ is the union of $x$ and all ascending curves with $x$ as their origin (resp. destination).
	We denote the ascending (resp. descending) manifold of $x$ by $A(x)$ (resp. $D(x)$).
\end{definition}

Our first result is Theorem~\ref{thm:open}.

\begin{theorem}\label{thm:open}
If $P \subset \Re^3$ is a nondegenerate convex polyhedron, then for any unstable point $y \in \partial P$, the descending manifold $D(s)$ is open in $\partial P$.
\end{theorem}

The proof of Theorem~\ref{thm:open} is based on Lemma~\ref{lem:stability}.

\begin{lemma}\label{lem:stability}
Let $P$ be nondegenerate, and let $q$ be a vertex of $P$ which is not an equilibrium point. Let $q'$ be a point on the ascending curve through $q$ such that $q'$ is not a vertex of $P$, and $[q,q']$ belongs to $\partial P$. Then, for every neighborhood $U$ of $q'$ in $\partial P$ there is a neighborhood $V \subset \partial P$ of $q$ such that the ascending curve through a point of $V$ intersects $U$.
\end{lemma}

\begin{proof}
First, we show that there is a unique candidate gradient at $q$. Let $F_1, \ldots, F_k$ denote the faces of $P$ containing $q$ in cyclic order. Let $E_i = F_{i-1} \cap F_{i}$ for $i=1,2,\ldots, k$. Let us call an edge $E_i$ \emph{descending} if the one-sided directional derivative at $q$ in the direction of $E_i$ is negative, and let us call it nondescending otherwise. We have seen in the proof of Theorem~\ref{gradient_existence} that if some $(\grad N_{<F_i>}) (q)$ is a candidate gradient, then $(\grad N_{<F_i>}) (q)$ is the extended gradient at $q$, and there is no candidate gradient generated by any other face $F_j$.
Now we show more, namely that apart from the face/edge generating the extended gradient, all nondescending edges are crossed edges. Before doing it, recall that by Remark~\ref{rem:connected_directions}, we may assume that the nondescending edges of $P$ at $q$ are $E_1, E_2, \ldots, E_j$ for some $1 \leq j < k$.

\emph{Case 1}: $(\grad^{ext} N_{\partial (P)}) (q) = (\grad N_{<F_i>}) (q)$ for some value of $i$. Then at least one of $E_i$ and $E_{i+1}$ is nondescending, and hence, we may assume that $1 \leq i \leq j$. First, observe that since $(\grad N_{<F_i>}) (q)$ points towards the interior of $F_i$, the lines through $E_i$ and $E_{i+1}$ separate this vector from the orthogonal projection of $o$ onto $<F_i>$. Thus, $E_i$ is a crossed edge from $F_{i-1}$ to $F_i$, and $E_{i+1}$ is a crossed edge from $F_{i+1}$ to $F_i$, respectively.

We show by induction on $t$ that $E_t$ is a crossed edge from $F_{t-1}$ to $F_t$, for $t=1,2,\ldots,i$. Note that we have already shown it for $t=i$. Assume that $E_{t+1}$ is a crossed edge from $F_t$ to $F_{t+1}$ for some $t \geq 1$. Let $o_t$ denote the orthogonal projection of $o$ onto $<F_t>$. The fact that $E_{t+1}$ is crossed from $F_t$ to $F_{t+1}$ implies that $<E_{t+1}>$ does not separate $o_t$ and $F_t$. Furthermore, the fact that both $E_{t+1}$ and $E_t$ are nondescending yields that the angle between $E_{t+1}$ and $[o_t,q]$, as well as the angle between $E_{t}$ and $[o_t,q]$, are both non-acute. But from this we have that $<E_t>$ separates $o_t$ and $F_t$, which yields that $E_t$ is crossed from $F_{t-1}$ to $F_{t}$ (see Figure~\ref{fig:crossed} for reference.). Repeating this argument for $i+1 \leq t \leq j$, we obtain that all nondescending edges at $q$ are crossed.

\begin{figure}[h]
  \centering
  \begin{tikzpicture}
    \draw[dashed] (150:2) -- (150:4) (0:0) -- (330:2) (180:2) -- (180:4) (0:0) -- (0:2);
    \draw (0:0) -- (120:2) -- (150:2) (0:0) -- (150:2) -- (180:2) (0:0) -- (180:2) -- (210:2) -- (0:0);
    \draw (60:0.2) node {$q$};
    \draw (150:4)++(60:0.2) node[rotate=-30, anchor=west] {$<E_{t+1}>$};
    \draw (180:4)++(90:0.2) node[anchor=west] {$<E_t>$};
    \draw (135:1.5) node {$F_{t+1}$};
    \draw (165:1.5) node {$F_t$};
    \draw (195:1.5) node {$F_{t-1}$};
    \fill (300:1) circle (0.05);
    \draw (300:1.3) node {$o_t$};
    \fill circle (0.05);
  \end{tikzpicture}
  \caption{A single step of the induction}\label{fig:crossed}
\end{figure}
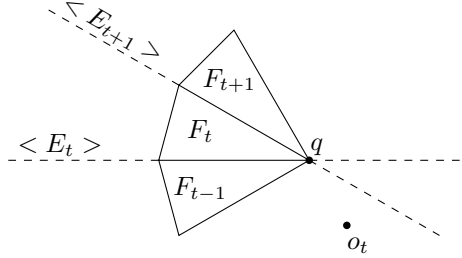

\emph{Case 2}: $(\grad^{ext} N_{\partial (P)}) (q) = (\grad N_{<E_i>}) (q)$ for some value of $i$. In this case a similar argument yields the desired statement.

We remark that the argument in Case 1 also shows that there are no two consecutive edges $E_t$, $E_{t+1}$ which are both descending, and $E_t$ is crossed from $F_t$ to $F_{t-1}$, and $E_{t+1}$ is crossed from $F_t$ to $F_{t+1}$. This observation, combined with the properties of nondescending edges proved above, yields the assertion.
\end{proof}

\begin{proof}[Proof of Theorem~\ref{thm:open}]
Let $c : \Re^- \to \partial P$ be an ascending curve with the unstable point $y$ as its destination, and let $q = c(\tau)$ be a point of $c$. We need to show that $q$ has a neighborhood $V$ such that for any point $q' \in V$, the destination of the ascending curve through $q'$ is $y$. By Theorem~\ref{thm:polygonal},  $c([\tau,0)) \cup \{ y \}$ is the union of finitely many closed segments such that
\begin{enumerate}
\item[(i)] each segment is a subset of a face or a followed edge of $P$
\item[(ii)] the endpoint of each segment, apart from $q$, is either a relative interior point of an edge, or a vertex of $P$.
\end{enumerate}
Thus, it is sufficient to show that if $[x_j,x_{j+1}]$ is one of these segments with $x_j$ preceding $x_{j+1}$ on $c$, then for every neighborhood $U$ of $x_{j+1}$
there is a neighborhood $V$ of $x_j$ such that for any $q' \in V$, the ascending curve through $q'$ intersects $U$. This statement is trivial if $x_j$ is a relative interior point of an edge. Assume that $x_j$ is a vertex, and let $x'$ be a relative interior point of $[x_j,x_{j+1}]$. Then, by the above observation, for any neighborhood $U$ of $x_j$ there is a neighborhood $V'$ of $x'$ such that any ascending curve intersecting $V'$ intersects $U$, and by Lemma~\ref{lem:stability}, the same statement holds for $x_j$ and $x'$ in place of $x'$ and $x_{j+1}$, respectively.
\end{proof}

\begin{definition}\label{defn:MorseSmale}
The intersection of an ascending and a descending polyhedral manifold, i.e. a set of the form $A(x) \cap D(y)$ for some (not necessarily distinct) equilibrium points $x,y$ of $P$, is called a \emph{cell} of the Morse\==Smale complex of $P$. The \emph{Morse\==Smale complex of $P$} is defined as the family consisting of its cells.
\end{definition}

\begin{definition}
  An ascending curve is \emph{isolated} if it does not belong to both a descending and an ascending polyhedral manifold.
\end{definition}

Recall that for generic smooth convex bodies, isolated integral curves correspond to integral curves starting or ending at saddle points of the body.
The following example shows that the same statement does not hold for every nondegenerate convex polyhedron $P$.

\begin{example}\label{ex:badguy}
Let $H^+$ and $H^-$ be two closed half planes in $\Re^3$ such that their intersection is the line $L=\{ (1,t,0) : t \in \Re \}$, they are symmetric to the $(x,y)$-plane, and their convex hull $P_0$ contains $o$. Let $s_1=(1,0,0)$, and let $H$ be a plane parallel to the $z$-axis such that it intersects the $x$-axis at a point $(1+\varepsilon,0,0)$ with some sufficiently small $\varepsilon > 0$, and passes through the point $q=(1,1,0)$. Let $s_2=(x,y,0)$ be a point of $H$ with $y > 1$, and let $H'$ be a plane containing the line through $s_2$, parallel to the $z$-axis, such that $s_2$ is a saddle point of the truncation of $P_0$ with $H$ and $H'$. Then the polygonal curve $[s_1,q] \cup [q,s_2]$ is an ascending curve from $s_1$ to $s_2$, corresponding to a saddle-saddle connection. Furthermore, if $o_+$ is the orthogonal projection of $o$ onto $H^+$, then $\conv \{ o_+, s_1, q \}$ is contained in $D(s_2)$.
\end{example}

\begin{figure}[h]
  \centering
  \begin{tikzpicture}
    \fill[lightgray] (0.5,0.5,0) -- (-1.5,3.5,-1) -- (-1.5,3.5,1) -- (0.5,-2.5,0) -- cycle;
    \draw (0.5,0.5,0) -- (-1.5,3.5,-1);
    \draw[densely dotted] (0,0,0) -- (0.5,0,0);
    \draw[-latex] (0.5,0,0) -- (3,0,0);
    \draw[densely dotted] (0,0,0) -- (0,1.3,0);
    \draw[-latex] (0,1.3,0) -- (0,3,0);
    \draw[densely dotted] (0,0,0) -- (0,0,0.5);
    \draw[-latex] (0,0,0.5) -- (0,0,3);
    \draw (0.5,0.5,0) -- (0.5,-2.5,0);
    \draw (0.5,0.5,0) -- (-1.5,3.5,1);
    \draw (-1.5,3.5,1) -- (0.5,-2.5,0);
    \draw (-1.5,3.5,-1) -- (-1.5,3.5,1);
    \draw[very thick] (0.5,0,0) -- (0.5,0.5,0) -- (-1.5,3.5,0);
    \fill (0.5,0,0) circle (0.05);
    \fill (-1.5,3.5,0) circle (0.05);
    \draw (0.7,-0.2,0) node {$s_1$};
    \fill (0.5,0.5,0) circle (0.05);
    \draw (0.7,0.5,0) node {$q$};
    \draw (-1.7,3.7,0) node {$s_2$};
    \fill (0.25,0,0.5) circle (0.05);
    \draw (0.5,-0.2,0.5) node {$o_+$};
    \draw (3,-0.2,0) node {$x$};
    \draw (0.2,3,0) node {$y$};
    \draw (0,-0.2,3) node {$z$};
  \end{tikzpicture}
  \caption{The polyhedron bounded by one possible choice of $H^+,H^-,H$ and $H'$ in Example~\ref{ex:badguy}.}
\end{figure}
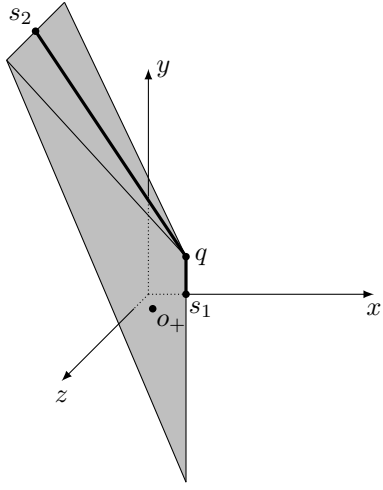

\begin{remark}\label{rem:whatisMS}
Clearly, for any nondegenerate convex polyhedron $P$ every saddle point $s$ is a relative interior point of a followed edge $E$, implying that all ascending curves ending at $s$ reach it in one of the two directions perpendicular to $E$, and all ascending curves starting at $s$ start along the edge $E$. Thus, for any nondegenerate convex polyhedron, the ascending and the descending manifolds at any saddle point cross.
\end{remark}

Based on Example~\ref{ex:badguy} and Remark~\ref{rem:whatisMS}, we introduce the following notion. Before that, we recall that a \emph{simplicial} convex polyhedron $P \subset \Re^3$ is defined as a convex polyhedron having only triangular faces, and observe that if $P = \conv \{ q_1, q_2, \ldots, q_k \}$ is a nondegenerate convex polyhedron with equilibrium points $s_1, s_2, \ldots, s_m$, and for any sufficiently small perturbation $P'= \conv \{ q'_1, q'_2, \ldots, q'_k \}$, with the point $q_i'$ being the perturbation of $q_i$, $P'$ is a nondegenerate convex polyhedron with $m$ equilibrium points $s'_1,s'_2,\ldots, s'_m$, and for a suitable choice of indices, $s'_j$ has the same type as that of $s_j$, and their distance is small.

\begin{definition}\label{defn:generic}
Let $P \subset \Re^3$ be a simplicial, nondegenerate convex polyhedron, and let the set of equilibrium points of $P$ be $\Eq(P)$. If, for any vertex $q$ of $P$ with $q \in A(x) \cap D(y)$ and $x,y \in \Eq(P)$, and any sufficiently small perturbation $P'$ of $P$, we have $q' \in A(x') \cap D(y')$, where $q'$ is the vertex of $P$ corresponding to $q$, and $x',y'$ are the equilibrium points of $P'$ corresponding to $x$ and $y$, respectively, then we say that $P$ is \emph{generic}, or \emph{Morse\==Smale}.
\end{definition}

\begin{theorem}\label{thm:generic}
Let $P \subset \Re^3$ be a generic polyhedron. Then the following holds:
\begin{enumerate}
\item[(i)] the descending manifold of any saddle point $s$ of $P$ contains no vertex of $P$;
\item[(ii)] the descending manifold of any saddle point consists of exactly two ascending curves;
\item[(iii)] the $2$-dimensional cells of the Morse\==Smale complex of $P$ are bounded by four ascending curves, and the four endpoints of these curves are a stable, a saddle, an unstable and a saddle point in this cyclic order.
\item[(iv)] the isolated ascending curves are exactly those starting or ending at a saddle point;
\end{enumerate}
\end{theorem}

\begin{proof}
Let $P$ be generic. Suppose for contradiction that there is an ascending curve $c$ of $P$ ending at a saddle point $s$ that contains a vertex. Let us denote the open segments of $c(\Re^-)$ by $S_1, S_2, \ldots, S_k$ in consecutive order. Let the endpoints of $S_i$ be $q_i$ and $q_{i+1}$. By the properties of saddle points, $S_k$ lies in the interior of a face of $P$. Thus, $q_k$ is  either a vertex or a relative interior point of a crossed edge. In the latter case $S_{k-1}$ lies in the interior of a face of $P$. Again, $q_{k-1}$ is either a vertex, or a relative interior point of a crossed edge of $P$. Repeating this argument we obtain that the last vertex $q_i$ of $P$ on $c$ has the property that for all $i < j \leq k$, the point $q_j$ is a relative interior point of a crossed edge of $P$. Now, let $F'$ be the face containing the segment $[q_i,q_{i+1}]$, and let $o'$ denote the orthogonal projection of $o$ onto $<F'>$. Let us perturb $q_i$ by moving it in the plane $<F'>$ to a point $q'_i$ such that $q_i$, $q_i'$ and $o'$ are not collinear. Then it is easy to see that the ascending curve through $q'_i$ on the perturbed polyhedron meets the edge of $s$ at a point different from $s$, and thus, it cannot end at $s$. But this contradicts our assumption that $P$ is generic, implying (i).

Now, assume that $P$ is a simplicial convex polyhedron with the property that the descending manifold of any saddle point $s$ contains no vertex other than $s$.
Note that any saddle point $s$ is a relative interior point of a followed edge. Thus, there are exactly two ascending curves starting at $s$, and they run on the edge of $s$ until they reach the endpoints of this edge, implying that any ascending curve starting at a saddle point contains a vertex of $P$. On the other hand, the ascending curves ending at $s$ reach it along the two segments in $\partial P$ perpendicular to the edge of $P$. By our assumption and the consideration in the previous paragraph, these segments belong to unique ascending curves whose all vertices lie on crossed edges of $P$. Thus, for any saddle point $s$ there are exactly two ascending curves starting at $s$, and these contain vertices of $P$, and exactly two ascending curves ending at $s$, and these do not contain vertices of $P$. This implies, in particular, that the other endpoints of these curves are not saddle points. Thus, the ascending curves starting at $s$ end at unstable, and the ones ending at $s$ start at stable points of $P$. From this, (ii) readily follows.

Note that since ascending curves cannot split, the two ascending curves ending at $s$ are necessarily disjoint, but the two curves starting at $s$ may merge later.
Nevertheless, the above facts are sufficient to show that, after removing all these curves from $\partial P$, any component is bounded by four arcs of such curves, whose endpoints are a stable, a saddle, an unstable and a saddle point in this cyclic order. Clearly, any such region $R$ is bounded by $4k$ arcs for some $k \in \mathbb{Z}^+$, and the sequence of the types of the endpoints of the corresponding curves are stable, saddle, unstable and saddle repeated $k$ times. Observe that by its definition, $R$ is an open subset of $\partial P$. On the other hand, by Theorem~\ref{thm:open} and since $R$ is open, $R \cap D(y)$ is an open subset of $R$ for any unstable point $y$ being the destination of an ascending curve in $\partial P$. Thus, the connectedness of $R$ yields that $k=1$ and (iii) holds, implying also (iv).
\end{proof}

\begin{corollary}
By Theorem~\ref{thm:generic}, the Morse\==Smale complex of a generic convex polyhedron $P$ consists of the following types of cells.
\begin{enumerate}
    \item The equilibrium points of $P$, called the \emph{vertices} of the Morse\==Smale complex.
    \item The closures of the isolated ascending curves, namely the curves starting or ending at a saddle point. These cells are called \emph{edges} of the complex.
    \item for a stable point $s$ and an unstable point $u$, the union of all ascending curves starting at $s$ and ending at $u$, and also the four isolated ascending curves intersecting the boundary of this union in an arc. These cells are called \emph{faces} of the complex.
\end{enumerate}
We note that a face of a Morse\==Smale complex is not necessarily homeomorphic to a disc. This happens if two ascending curves in the boundary of the face merge. The existence of such a face in a Morse\==Smale complex is equivalent to the property that the complex is not a CW-decomposition of $\bd P$.
\end{corollary}

\begin{definition}\label{defn:MSgraph}
Let $P \subset \Re^3$ be a generic convex polyhedron. The topological graph $G$ on $\partial P$ whose vertices are the equilibrium points of $P$, and whose edges are the isolated ascending curves of $P$, is called the \emph{Morse\==Smale graph} generated by $P$. This graph can be regarded as a \emph{$3$-colored} quadrangulation of $\Sph^2$, where the `colors' of the vertices are the three types of a critical point.
\end{definition}

\section{An algorithm to compute Morse\==Smale complex}\label{sec:algorithm}

It is a natural problem to determine the Morse\==Smale graph $G$ of a generic convex polyhedron $P$.
Our results in Sections~\ref{subsec:local} and \ref{subsec:global} provide an algorithmic answer for this problem:

\medskip

\begin{enumerate}
\item[Step 1]: Find the orthogonal projections of $o$ onto the line of each edge of $P$.
\item[Step 2]: Find which edge is followed, which edge is crossed.
\item[Step 3]: Find the saddle points of $P$.
\item[Step 4]: For each saddle point, find the ascending curves ending there. Note that by genericity, there are exactly two such curves for each saddle point, and they intersect only crossed edges.
\item[Step 5]: For each saddle point, find the ascending curves starting there.
\end{enumerate}

\medskip

We implemented the above algorithm and we illustrate it on the following example. 
Let $P_\text{ex}$ be a polyhedron constructed as follows:
\begin{enumerate}
\item  Consider a regular octagon $O_1 \subset \Re^3$ with diameter $d=2$, lying in a plane parallel to the $(x,y)$ plane and with center $(0,0,1)$. Denote the vertices of $O_1$ by $v_i$, $(i=1,2, \ldots, 8)$.

\item Project $O_1$ orthogonally onto the plane $z=-1$ to obtain $O_2$.

\item Rotate $O_2$ by $\varphi=\frac{\pi}{20}$ clockwise and denote the vertices of the rotated octagon by $v_i$, $(i=9,10, \ldots, 16)$.

\item Add the points $v_0=(0,0,1.2), v_{17}=(0,0,-1.2)$.

\item  Define $P_\text{ex}$ as the convex hull of the points $\{v_i\}$, $(i=0,1, \ldots, 17)$.

\end{enumerate}

To equip $P_\text{ex}$ with a radial distance function, we select the point $o_\text{ex}=(0.5,0.5,0.5)$ as origin (in this case the origin, as a center of mass, corresponds to inhomogeneous material distribution).  Figure~\ref{fig:example} shows the ascending curves of $P_\text{ex}$ and we can see that $P_\text{ex}$ is generic.

\begin{figure*}[h]
    \centering
    \includegraphics[width=.7\textwidth]{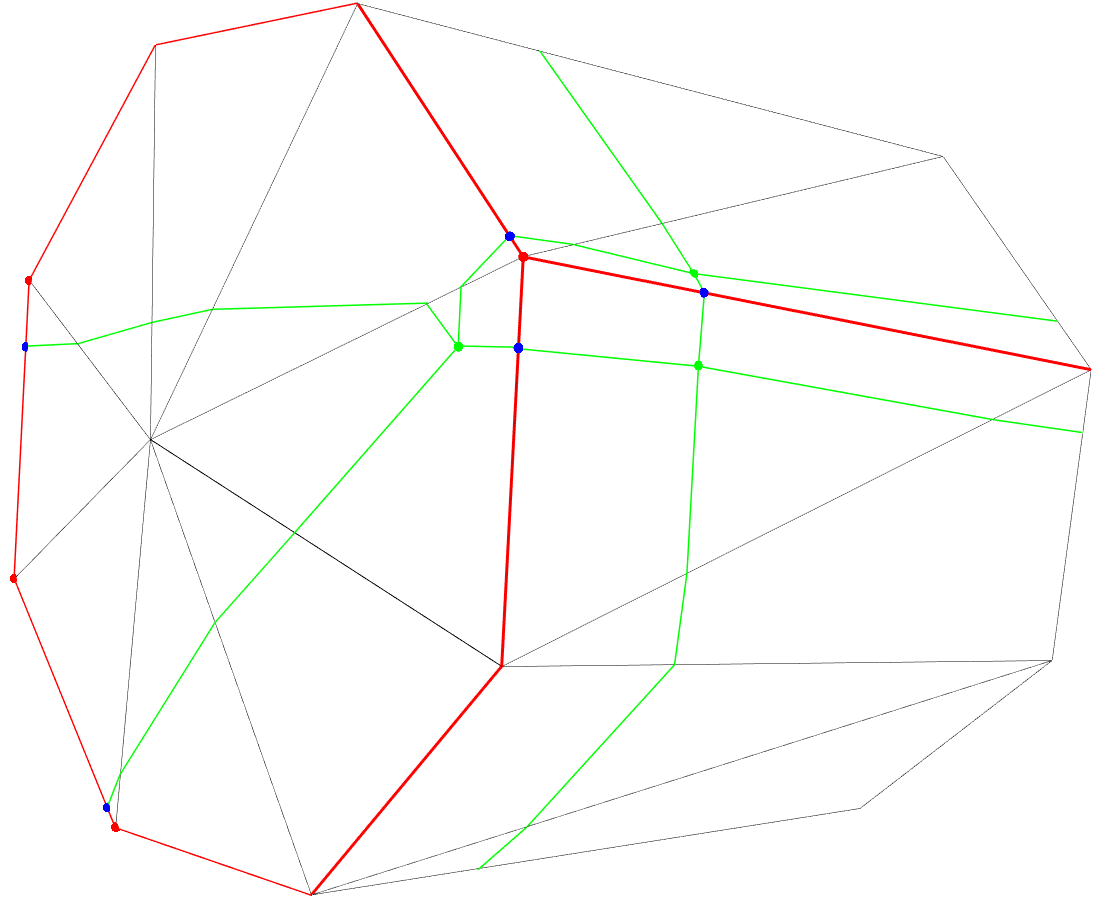}
    \caption{Isolated ascending curves on $P_\text{ex}$ with respect to $o_\text{ex}$. Red, green and blue points represent unstable, stable and saddle\=/type equilibria, respectively. Stable\=/saddle isolated ascending curves are shown in green, saddle\=/unstable isolated ascending curves are shown in red.}\label{fig:example}
\end{figure*}

\section{Concluding remarks and summary}\label{sec:remarks}

Before summarizing our results we make some general remarks and state an open question.

\subsection{Some remarks and a question}

\begin{remark}
We have seen in Theorem~\ref{thm:open} that for any nondegenerate convex polyhedron, the descending manifold of any unstable point is open. Nevertheless, there are even generic convex polyhedra with stable points whose ascending manifolds are not open. This happens, for example, if two ascending curves in the boundary of a $2$-cell of the Morse\==Smale complex merge. Figure~\ref{fig:open} shows a face of $P$ where this situation can arise.
\end{remark}

\begin{figure}[h]
  \centering
  \begin{tikzpicture}
    \draw (-30:2) -- (90:2) -- (210:2) -- cycle;
    \fill (0:0) circle (0.1);
    \fill (30:1) circle (0.1);
    \fill (150:1) circle (0.1);
    \fill (90:2) circle (0.1);
    \draw[very thick, dashed] (0:0) -- (30:1) (0:0) -- (150:1) (30:1) -- (90:2) (150:1) -- (90:2);
    \draw (-90:0.3) node {$s$};
    \draw (30:1.4) node {$h_1$};
    \draw (150:1.4) node {$h_2$};
    \draw (90:2.3) node {$q$};
  \end{tikzpicture}
  \caption{The saddle points $h_1$ and $h_2$ are destinations of two isolated ascending curves from the stable point $s$. They are also origins of isolated ascending curves passing through the vertex $q$. If $q$ is unstable then the $2$-cell bounded by the previously mentioned isolated ascending curves is open, otherwise it is not.}\label{fig:open}
\end{figure}
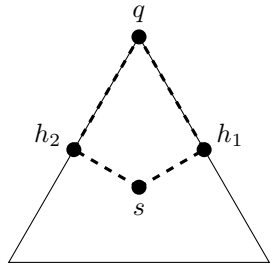

Before our next remark, let us recall a few concepts from Baire category theory. A topological space is called a \emph{Baire space}, if the union of every countably many closed sets with empty interior has empty interior. In particular, the family of convex polyhedra with at most $n$ vertices is a complete metric space, and thus, it is a Baire space \cite{schneider}. A set in a complete metric space is called \emph{typical} or \emph{residual} if its complement is a countable union of nowhere dense sets. In particular, if a set in a complete metric space is open and everywhere dense, it is typical.

\begin{remark}
Let $\mathcal{P}_m$ denote the family of convex polyhedra in $\Re^3$ with at most $m$ vertices. Observe that  $\mathcal{P}_m$ is a complete metric space with respect to Hausdorff distance, and thus, it is a Baire space. Furthermore, a sufficiently small perturbation of the vertices of a generic convex polyhedron with $m$ vertices yields a generic polyhedron with $m$ vertices, and any convex polyhedron with at most $m$ vertices can be approximated arbitrarily well by such a polyhedron. Thus, in Baire category sense, a typical convex polyhedron in  $\mathcal{P}_m$ with at most $m$ vertices is a generic convex polyhedron with $m$ vertices.
\end{remark}

\begin{remark}
The computation of Morse\==Smale complexes is a fast growing area in data visualization and some algorithms have even
 benefited from GPU acceleration \cite{parallel}. Apparently, the speed of such algorithms is essential, we make some
simple observations about our algorithm.
Note that if $P$ has $n$ vertices, it has at most $2n-4$ faces and $3n-6$ edges, and each vertex, edge, face contains at most one equilibrium point.
Thus, Steps 1-3 in our algorithm can be carried out in $O(n)$ time, and by Theorem~\ref{thm:polygonal}, Steps 4-5 can be carried out in $O(n^2)$ time. 
Clearly, at least $\Omega(n)$ steps are necessary to find the Morse\==Smale graph of $P$.  While the exact magnitude of a fastest algorithm that solves this problem is an open question (see below), we note that on all test examples (with a couple hundred vertices per polyhedron on average) our algorithm, running on a
single core of a standard laptop CPU, computed the Morse\==Smale complex well under a second.
\end{remark}

\begin{problem}
Find the exact magnitude of a fastest algorithm that finds the Morse\==Smale graph of any generic convex polyhedron. 
\end{problem}

  \begin{remark}\label{rem:flow_cmp}
Let $X \in \Re^3$ be a finite point set, and for any point $q \in \Re^3$, set $h : \Re^3 \to \Re$, $h(q) = \min \{ ||q-x||^2 : x \in X \}$. This function is a piecewise smooth function which is not smooth exactly at the boundary points of the Voronoi cells of the points of $X$. In \cite{giesen2003flow}, the authors considered the gradient vector field in $\Re^3$ defined by $h$, and examined the properties of the $3$-dimensional Morse\==Smale complex associated to it. An important tool of their decription was a geometric interpretation of the direction of the steepest ascent, which can be stated as follows: Consider a point $q \in \Re^3$. Then there is a unique smallest dimensional Voronoi object $V$ (vertex, edge, face or cell) in the Voronoi decomposition of $\Re^3$ by $X$ which contains $q$. Let $D$ be Delaunay object in the Delaunay decomposition of $\conv (X)$ by $X$ dual to $V$, and let $x$ be the point of $D$ closest to $q$, which is called the \emph{driver} of $q$. Then the direction of the steepest ascent of $h$ at $q$ is the direction of the vector $q-x$. In the point set problem mentioned in the introduction, these directions, viewed as a (not everywhere smooth) vector field in $\Re^3$, are used to define ascending curves through noncritical points, which is then applied to define the flow complex associated to the point set.

The results in \cite{giesen2003flow} are especially interesting if we recall that if $P$ is a convex polyhedron in $\Re^3$ with $x \in \inter (P)$, then $P$ is the Voronoi cell of the point set $X$ consisting of $x$ and its reflected copies to all face planes of $P$. Furthermore, for any $x \in \Re^3$ the gradient vector field defined by $q \mapsto ||q-x||$ coincides with the normalized gradient vector field defined by $q \mapsto ||q-x||^2$. On the other hand, an important difference between this model and our setting is that in the point set problem, the direction of steepest ascent is considered among all directions in $\Re^3$, whereas in our paper it is restricted to $\partial P$.

We illustrate this difference by the following example (see Figure~\ref{not_voronoi}) which shows that in the flow complex defined by $P$ in the previous paragraph, an ascending curve may leave $\partial P$. First, let the reference point in $P$ be the origin $o$. Let $W$ be a (convex) wedge bounded by the half planes $H_1$ and $H_2$ perpendicular to the $(x,y)$-plane such that $W$ is symmetric to the $(y,z)$-plane and it contains $o$ in its interior. For $i=1,2$, let $o_i$ denote the reflection of $o$ to the plane containing $H_i$. These points are in the $(x,y)$-plane. Let $L=H_1 \cap H_2$. Then, with respect to the point set $\{ o,o_1, o_2\}$, $L$ is a Voronoi edge whose associated Delaunay face is the triangle with vertices $o,o_1,o_2$. Thus, if $d$ denotes the midpoint of the segment $[o_1,o_2]$, then for any $q \in L$ the direction of the steepest ascent points in the direction of $q-d$ from $q$, and thus, it leaves the wedge $W$. The direction of steepest ascent would coincide with $L$ iff $L$ intersects the triangle with vertices $o,o_1,o_2$. Now, we construct $P$ by truncating $W$ with planes in such a way that the obtained convex polyhedron $P$ has an edge lying in $L$, providing the desired example.
\end{remark}

\begin{figure}[h]
  \centering
  \begin{tikzpicture}
    \draw[dashed] (2,0,0) -- (-1,0,1.5) -- (-1,0,-1.5) -- cycle;
    \draw[loosely dotted] (0,0,0) -- (2,0,0);
    \draw[-Latex] (2,0,0) -- (3,0,0);
    \draw (3.2,0,0) node {$y$};
    \draw[loosely dotted] (0,0,0) -- (0,1.5,0);
    \draw[-Latex] (0,1.5,0) -- (0,3,0);
    \draw (0,3.2,0) node {$z$};
    \draw[loosely dotted] (0,0,0) -- (0,0,1);
    \draw[-Latex] (0,0,1) -- (0,0,3);
    \draw (0,0,3.3) node {$x$};
    \draw (0.9,1.5,0) node {$W$};
    \draw[very thick, dashed] (2,-3,0) -- (2,3,0);
    \draw (2.2,3,0) node {$L$};
    \fill (0,0,0) circle (0.1);
    \draw (-0.3,0,0) node {$o$};
    \fill (1,0,2) circle (0.1);
    \draw (1.3,-0.1,2) node {$o_1$};
    \fill (1,0,-2) circle (0.1);
    \draw (0.7,0.1,-2) node {$o_2$};
    \draw[densely dotted] (1,0,2) -- (1,0,-2) -- (0,0,0) -- cycle;
    \fill (1,0,0) circle (0.05);
    \draw (0.8,0,0) node {$d$};
    \draw (2,1.5,0) -- (-1,1.5,1.5) -- (-1,1.5,-1.5) -- cycle;
    \draw (2,-1.5,0) -- (-1,-1.5,1.5);
    \draw (-1,1.5,1.5) -- (-1,-1.5,1.5);
    \draw (2,1.5,0) -- (2,-1.5,0);
    \foreach \y in {-1.6,-1.7,...,-2.5} {
      \draw[gray] (2,\y,0) -- (1,\y,0.5);
    }
    \fill[white] (1.5,-2,0.25) circle (0.3);
    \draw (1.5,-2,0.25) node {$H_1$};
    \foreach \y in {1.6,1.7,...,2.5} {
      \draw[gray] (2,\y,0) -- (1,\y,-0.5);
    }
    \fill[white] (1.5,2,-0.25) circle (0.3);
    \draw (1.5,2,-0.25) node {$H_2$};
    \fill (2,-1,0) circle (0.08);
    \draw (1.7,-1,0) node {$q$};
    \draw[-Latex] (2,-1,0) -- (2.5,-1.6,0);
    \draw (4,-1.5,0) node {direction of $q-d$};
    \draw[densely dotted] (1,0,0) -- (2,-1,0);
  \end{tikzpicture}
  \caption{The wedge $W$ of Remark~\ref{rem:flow_cmp} truncated by two planes parallel to the $(x,y)$-plane and one parallel to the $(x,z)$-plane, forming the polyhedron $P$. $P$'s intersection with the $(x,y)$-plane is shown with dashed lines.}\label{not_voronoi}
\end{figure}

\subsection{Summary}

Motivated by geophysical applications, in this paper we extended Morse\==Smale theory to a convex polyhedron $P$, interpreted as a scalar
(radial) distance measured from a point $o$ in its interior. We defined the polyhedral Morse\==Smale complex
and the Morse\==Smale graph $G$ associated with $P$ and $o$ and we built an algorithm to compute $G$. In geophysical applications, $o$ is the center of mass of the particle
and the scanned image of the particle is a polyhedron, so this algorithm associates one topological graph to a scanned
particle.

In classical Morse\==Smale theory the Morse\==Smale complex is defined by a function $f: \mathcal{M} \to \Re$, acting
on the compact manifold $\mathcal{M}$ and in the classical case $f$ is considered to be smooth. 
From the abstract point of view, the natural extension of this theory is
to extend the class of functions $f$ to be considered and this path was followed in \cite{banchoff_polyhedral} and \cite{edelsbrunner_morse_smale}
where the authors considered piecewise linear functions. Here we took a slightly different approach: instead of defining
a new class of functions we would discuss, we considered the geometric image of the function as a given object.
In our case, this image was a convex polyhedron and the function $f$ was defined on the sphere as a radial
distance, measured from a point. We hope that this
geometric approach may lead to other potential generalizations of Morse\==Smale theory.

\section*{Funding, competing interests, and data availability statement} 

\noindent
\textbf{Funding}: The research reported in this paper was supported by the BME Water Sciences \& Disaster Prevention TKP2020 Institution Excellence Subprogram, grant no. TKP2020 BME-IKA-VIZ, the NKFIH grant K134199, the J\'anos Bolyai Research Scholarship of the Hungarian Academy of Sciences, and the \'UNKP-20-5 New National Excellence Program by the Ministry of Innovation and Technology.

\noindent
\textbf{Competing interests}: The authors have no competing interests to declare that are relevant to the content of this article.

\noindent
\textbf{Data availability statement}: In this paper we used computer generated datasets to illustrate our findings, however, none of the results presented in the paper do rely on any of the computer generated datasets.
The datasets generated and/or analysed during the current study are available from the corresponding author on reasonable request.

\bibliographystyle{plain}
\bibliography{bib}

\end{document}